\documentclass[11pt,a4paper, final, twoside]{article}
\usepackage{amsmath}
\usepackage{fancyhdr}
\usepackage{amsthm}
\usepackage{amsfonts}
\usepackage{amssymb}
\usepackage{amscd}
\usepackage{latexsym}
\usepackage{amsthm}
\usepackage{graphicx}
\usepackage{graphics}
\usepackage{natbib}
\usepackage{afterpage}
\usepackage[colorlinks=true, urlcolor=blue,  linkcolor=black, citecolor=black]{hyperref}
\usepackage{color}
\newtheorem{thm}{Theorem}[section]
\newtheorem{algorithm}[thm]{Algorithm}

\newtheorem{problem}[thm]{Problem}
\theoremstyle{proposition}
\newtheorem{prop}{Proposition}[section]

\theoremstyle{definition}
\newtheorem{defn}{Definition}[section]
\theoremstyle{remark}

\numberwithin{equation}{section}
\begin{document}

\begin{flushright}

{\Large \textbf{\\On an application of multidimensional arrays}}\\[5mm]
{\large \textbf{Krasimir Yordzhev$^\mathrm{*}$\footnote{\emph{E-mail: yordzhev@swu.bg}}  }}\\[1mm]
{\footnotesize \it South-West University ''Neofit Rilski'',\\ Blagoevgrad, Bulgaria}
\end{flushright}

{\Large \textbf{Abstract}}\\[4mm]
\fbox{%
\begin{minipage}{5.4in}{\footnotesize This article discusses some difficulties in the implementation of combinatorial algorithms associated with the choice of all elements with certain properties among the elements of a set with great cardinality.The problem has been resolved by using multidimensional arrays. Illustration of the method is a solution of the problem of obtaining one representative from each equivalence class with respect to the described in the article equivalence relation in the set of all $m\sim n$  binary matrices. This equivalence relation has an application in the mathematical modeling in the textile industry.
} \end{minipage}}\\[1mm]

\section{Introduction and task formulation}\label{I1}
The following problem often occurs in computer science:
\begin{problem}\label{problem1}
Let $M$  be a finite set and let $\sim$  be an equivalence relation in $M$. Describe and implement an algorithm that receives exactly one representative from each equivalence class with respect to  $\sim$.
\end{problem}

As a consequence of this problem follows the combinatorial problem of finding the cardinality of the factor set $\widetilde{M} = M_{/\sim}$   consisting of all equivalence classes of $M$  with respect of  $\sim$.

We assume that for every  $x\in M$, there is a procedure $K(x)$  which receives all elements of  $M$, which are equivalent to  $x$.

Since $M$  is a finite set, then there exists bijective mapping
$$b\; :\; \leftrightarrow \left\{ 1,2,\ldots ,\left| M\right| \right\} ,$$
which will call \emph{numbering function}. Thus, each element of  $M$ uniquely corresponds to an element of Boolean array $H[\; ]$  with size equal to the cardinality $|M|$  of the set  $M$. Moreover, the element  $x\in M$ is \emph{selected} if $H[b(x)]=1$  and $x$  is \emph{not selected} if  $H[b(x)]=0$.

The next algorithm is a modification of the well-known method, known as ''Sieve of Era\-to\-sthe\-nes'' [\cite{Reingold,markovska}] solves Problem \ref{problem1}.

\begin{algorithm}\label{Algorithm1}
Receives exactly one representative of each equivalence class of the factor-set  $\widetilde{M}=M_{/\sim}$.

\textbf{Input}: Finite set $M$

\textbf{Output}: Set $N\subseteq M$

\begin{enumerate}
\item	 $N:=\emptyset$;

\item	Declare a Boolean array $H[\; ]$  with size equal to the cardinality $|M|$  of the set $M$  and put $H[b(x)]:=0$  for all $x\in M$;

\item	For every $x\in M$  such that $H[b(x)]=0$\textbf{ do}

\hspace{1cm}\{ \textbf{Begin} of loop 1

\item	\hspace{1cm}$N:=N\cup \{x\}$;

\item	\hspace{1cm}$H[b(x)]:=1$;

\item\label{6}	\hspace{1cm}Using the procedure $K(x)$  obtain the set $P_x =\{y\in M \; |\;y\sim x \}$;

\item	\hspace{1cm}For every $y\in P_x$  obtained in step \ref{6} \textbf{do}

\hspace{2cm} \{ \textbf{Begin} of loop 2

\item \hspace{2cm} $H[b(y)]:=1$;

\hspace{2cm} \textbf{End} of loop 2 \}

\hspace{1cm}\textbf{End} of loop 1 \}

\item	End of the algorithm.
\end{enumerate}
\end{algorithm}

Algorithm \ref{Algorithm1} has a number of disadvantages, the main of which is that it is practically inapplicable for programs when a sufficiently great number of elements is present in the base set $M$. This limitation comes from the maximum integer, which can be used in the corresponding programming environment. For example, by standard in the C++ language the biggest number of the type \textbf{unsigned long int} is equal to $2^{32} - 1$, which in a number of cases is insufficient for the previously defined array $H[\; ]$  to be completely addressed. The purpose of this article is to avoid this problem by using a multidimensional Boolean array, the elements of which have a one-to-one correspondence to the elements of the base set, with a much smaller range of the indices. There are many publications related to multidimensional arrays, for example [\cite{arxiv}], but they are not used for our specific goals and objectives. Another solution to the problem is the use of dynamic data structures or other special programming techniques [\cite{Collins,Sutter,Tan}] but it is not the subject of consideration in this article. 	

\emph{Binary (or Boolean, or (0,1)-matrix)} is a matrix whose elements are equal to 0 or 1.

Let ${\cal B}_{m\times n}$  be the set of all $m\times n$  binary matrices. It is well known that

\begin{equation}\label{eq1}
\left| {\cal B}_{m\times n} \right| = 2^{mn}
\end{equation}

In this work, we will consider and solve the following special case of Problem \ref{problem1}:

\begin{problem}\label{problem2}
 Let ${\cal B}_{m\times n}$  be the set of all $m\times n$  binary matrices and let  $X,Y\in {\cal B}_{m\times n}$. We define an equivalence relation $\rho$  as follows: $X\rho Y$ if and only if we can obtain $X$  from $Y$  by a sequential moving of the last row or column to the first place. Find the cardinality $|{\cal B}_{m\times n/\rho} |$  of the factor-set $\widetilde{M} ={\cal B}_{m\times n/\rho}$  and receive a single representative of each equivalence class.
\end{problem}

The proof that $\rho$  is an equivalence relation is trivial and we will omit it here.

The equivalence classes of ${\cal B}_{m\times n}$   by the equivalence relation $\rho$  are a particular kind of \emph{double coset} [\cite{r2,r4,r6}]. They make use of substitutions group theory and linear representation of finite group theory [\cite{r4,r6}].

When $m=n$, the elements of the factor-set $\widetilde{M} ={\cal B}_{n\times n/\rho}$  put carry into practice in the textile technology [\cite{b2,textile}].

In [\cite{y9}] an algorithm is shown, which utilizes theoretical graphical methods for finding the factor set  $\widetilde{S} =S_{n/\rho}$, where $S_{n} \subset {\cal B}_{n\times n}$  is a set of all permutation matrices, i.e. binary matrices having exactly one 1 on each row and each column. In [\cite{az}] we extended this problem in the case when $\rho$  is an arbitrary permutation.

The author of this paper is not familiar with an existing a general formula expressed as a function of $m$ and $n$ for finding  $|{\cal B}_{m\times n/\rho} |$.  The goal of this paper is to describe an effective algorithm for finding the number of elements of the factor set  $\widetilde{M} ={\cal B}_{m\times n/\rho}$, as well as finding a single representative of each equivalence class. Here we will describe an algorithm, which overcomes some difficulties, which would inevitably arise with sufficiently great m and n if we apply the classical algorithm (Algorithm \ref{Algorithm1}). The main difficulty arises from the great number of elements of $\widetilde{M} ={\cal B}_{m\times n/\rho}$  with comparatively small integers $m$  and $n$, according to formula (\ref{eq1}).

For undefined notions and definitions, we refer to [\cite{r1,sachkov}].

\section{Description of an algorithm with the use of a multi\-di\-men\-si\-o\-nal array}\label{I3}
\begin{thm} \label{Theorem1} Let us denote by ${\cal P}_n$  the set
\begin{equation}\label{(11)}
{\cal P}_n =\left\{ 0,1,\ldots ,2^n -1 \right\}
\end{equation}
  Then a one-to-one correspondence (bijection) between the elements of the Cartesian product
  $\displaystyle {\cal P}_n^m =\underbrace{{\cal P}_n \times {\cal P}_n \times \cdots \times {\cal P}_n}_m $  and the elements of the set ${B}_{m\times n}$  of all  $m\times n$  binary matrices exists.
\end{thm}

\begin{proof}
We consider the mapping  $\alpha  : {\cal P}_n^m
\to {\cal B}_{m\times n}$, defined in the following way: If  $\pi\in {\cal P}_n^m$ and $\pi =<p_1 ,p_2 ,\ldots ,p_m >$  then let us denote by  $z_i$,  $i=1,2,\ldots ,m$, the representation of the integer  $p_i$ in a binary notation, and if less than $n$ digits 0 or 1 are necessary, we fill $z_i$ from the left with insignificant zeros, so that $z_i$  will be written with exactly $n$ digits. Since by definition,  $p_i \in {\cal P}_n $, i.e.  $0\le p_i \le 2^n -1$, this will always be possible. Then we form an $m\times n$   binary matrix, so that the $i$-th row is  $z_i ,$ $i=1,2,\ldots m$. Apparently this is a correctly defined mapping of  ${\cal P}_n^m$ to ${\cal B}_{m\times n}$. It is clear that for different $n$-tuples from ${\cal P}_n^m$  with the help of $\alpha$  we will obtain different matrices from  ${\cal B}_{m\times n}$, i.e. $\alpha$  is an injection. Conversely, rows of each binary matrix can be considered as natural numbers, written in binary system by using exactly $n$ digits 0 or 1, eventually with insignificant zeros in the beginning, that is, these numbers belong to the set  ${\cal P}_n =\{ 0,1,\ldots ,2^n -1\}$. Therefore each  $m\times n$ Binary matrix corresponds to an $m$-tuple of numbers   $<p_1 ,p_2 ,\ldots ,p_m >\in {\cal P}_m^n$, that is,  $\alpha$ is a surjection. Hence  $\alpha$ is a bijection.
\end{proof}

It is easy to see the validity of the following statement, which in fact shows the meaning of our considerations.

\begin{prop}\label{propmu}
Let us denote by  $\mu$ the maximum integer, which we use when coding the elements of the set ${\cal B}_{m\times n}$  by means of the bijection, defined in Theorem \ref{Theorem1}. Then, for sufficiently great $m$ and $n$, the following is valid:
\begin{equation}\label{(12)}
\mu =\max \left( 2^n -1,m\right) \ll |{\cal B}_{m\times n}
|=2^{mn}
\end{equation}	
\end{prop}
\begin{proof} Trivial. \end{proof}

Let $a$ and $b$  be integers,  $b\ne 0$. With $a/b$ we will denote the operation ''integer division'' of $a$ by $b$, i.e. if the division has a remainder, then the fractional part is cut, and with  $a\% b$ we will denote the remainder when dividing $a$  by  $b$. In other words, if $\displaystyle \frac{a}{b} = p+\frac{q}{b}$, where $p$ and $q$ are integers, $0\le q< b$ then by definition $a/b =p$, $a\% b =q$.

We consider the function	

\begin{equation}\label{(13)}
\xi (a) =\left( a\% 2\right) 2^{n-1} +a/2 ,
\end{equation}
where $\%$  and  $/$  are the defined in the above operations.

\begin{defn}\label{Definition1}
 Let $\alpha$   be the defined in the proof of Theorem \ref{Theorem1} bijection and let the functions $f_r ,f_c : {\cal P}_n^m \to
{\cal P}_n^m$  be defined such that for every $\pi =<p_1 ,p_2 ,\ldots ,p_m > \in {\cal P}_n^m$
\begin{equation}\label{(15)}
f_r (\pi )=<p_m ,p_1 ,p_2 ,\ldots p_{m-1} >
\end{equation}	
\begin{equation}\label{(16)}
f_c (\pi )=< \xi (p_1 ), \xi (p_2 ),\ldots ,\xi (p_m )>,
\end{equation}
where the function  $\xi (a) $ is the defined with (\ref{(13)}).
\end{defn}

\begin{thm}\label{Theorem2}
Let $A\in {\cal B}_{m\times n}$  be an arbitrary $m\times n$  binary matrix and let $\alpha$ be the defined in the proof of Theorem \ref{Theorem1} bijection. Let us to get the matrices
\begin{equation}\label{(17)}
B=\alpha \left( f_r \left( \alpha^{-1} (A)\right)\right)
\end{equation}	
and				
\begin{equation}\label{(18)}
C=\alpha \left( f_c \left( \alpha^{-1} (A)\right)\right)
\end{equation}	

Then $B$ is obtained from $A$ by moving the last row to the first place, and $C$ is obtained from $A$ by moving the last column to the first place (respectively the first row or column becomes the second, the second becomes the third respectively etc.).
\end{thm}

\begin{proof} Let  $\pi =<p_1 ,p_2 ,\ldots ,p_m >=\alpha^{-1} (A)\in {\cal P}_n^m$. Then the integer  $p_i$, $0\le p_i \le 2^n -1$, $i=1,2,\ldots ,m$  will correspond to the $i$-th row of the matrix $A$. Then obviously, the matrix $B=\alpha (f_r (<p_1 ,p_2 ,\ldots ,p_m >))=\alpha (<p_m ,p_1 ,p_2 ,\ldots ,p_{m-1} >)$  is obtained from $A$ by moving the last row in the place of the first one, and moving the remaining rows one row below.

Let $p_i \in {\cal P}_n =\left\{ 0,1,\ldots ,2^n -1 \right\}$, $i=1,2,\ldots ,m$. Then $d_i =p_i \% 2$  gives the last digit of the binary notation of the integer  $p_i$. If  $p_i$ is written in binary notation with precisely $n$  digits, optionally with insignificant zeros in the beginning, then by applying integer division of  $p_i$   by 2, we practically remove the last digit $d_i$  and we move it to the first position, in case we multiply by $2^{n-1}$  and add it to  $p_i /2$. This is, by definition, how the function  $\xi (p_i )$ works. Hence, the $m\times n$ binary matrix $C=\alpha (f_c (<p_1 ,p_2 ,\ldots , p_m >))=\alpha (<\xi (p_1 ),\xi (p_2 ),\ldots ,\xi (p_m )>)) $  is obtained from the matrix $A$ by moving the last column to the first position, and all the other columns are moved one column to the right.
\end{proof}

From the definitions of the functions  $f_r$, according to (\ref{(15)}) and  $f_c$, according to (\ref{(16)}) it is easy to verify the validity of the following
\begin{prop} \label{Proposition2}
If by definition
\begin{equation}\label{(19)}
f_r^0 (\pi )=f_c^0 (\pi )=\pi
\end{equation}		
\begin{equation}\label{(20)}
f_r^k (\pi )=f_r \left( f_r^{k-1} (\pi )\right)
\end{equation}		
\begin{equation}\label{(21)}
f_c^k (\pi ) =f_c \left( f_c^{k-1} (\pi )\right) ,
\end{equation}
where $\pi\in {\cal P}_n^m$ and $k$ is a positive integer, then
\begin{equation}\label{(22)}
f_r^m (\pi )=\pi
\end{equation}		
and
\begin{equation}\label{(23)}
f_c^n (\pi )=\pi .
\end{equation}		
\end{prop}

\begin{proof} Trivial.
\end{proof}

As a direct consequence of Theorem \ref{Theorem1}, Theorem \ref{Theorem2}, Proposition \ref{Proposition2} and their constructive proofs, it follows that the following algorithm that finds exactly one representative of each equivalence class with respect to the defined in Problem \ref{problem2} equivalence relation  $\rho$   and that calculates the cardinality of the factor set  ${{\cal B}_{m\times n}}_{/\rho}$.

\begin{algorithm}\label{Algorithm2} Receives exactly one representative of each equivalence class of the factor-set  $\widetilde{M}=M_{/\rho}$ and calculates the cardinality of the factor set $\widetilde{M}=M_{/\rho}$  when $m$ and $n$ are given.

\begin{enumerate}
\item We declare the $m$-dimensional Boolean arrays $W1$ and $W2$ which we will be indexed by using the elements of the set  ${\cal P}_n^m$, i.e. $W1[<p_1 ,p_2 ,\ldots ,p_m >]$  will correspond to the element  $<p_1 ,p_2 ,\ldots ,p_m > \in{\cal P}_n^m$. We proceed analogically with the array  $W2$.
	
\item Initially we take all elements of $W1$ and $W2$  to be 0. In $W1$  we will remember all elements selected from ${\cal B}_{m\times n}$  (one for each equivalence class) by changing $W1[<p_1 ,p_2 ,\ldots ,p_m >]$ to 1 if we have selected the element $\alpha ( <p_1 ,p_2 ,\ldots ,p_m >)$  for a representative of the respective equivalence class. We will change the elements of  $W2$ to 1 for each selection of an element from   ${\cal B}_{m\times n}$, i.e. for each  $\pi'' \in {\cal P}_n^m$, for which there exists  $\pi' \in {\cal P}_n^m$, such that  $W1[\pi' ]=1$  and  $\alpha (\pi'' )\rho \alpha (\pi' )$, or in other words, $\pi'$ and $\pi''$ encode  two different matrices of the same equivalence class as we have chosen $\alpha (\pi' )$  for a representative of this equivalence class.

\item We declare the counter  $N$, which we initialize by 0. In case of normal ending of the algorithm, $N$  will be showing the cardinality of the factor set  ${{\cal B}_{m\times n}}_{/\rho}$.

\item While a zero element exists in в $W2$  \textbf{do}

\hspace{1cm}	\{ \textbf{Begin} of loop 1

\item \hspace{1cm} We choose the minimal $\pi =<p_1 ,\pi_2 ,\ldots ,\pi_m >\in {\cal P}_n^m$  according to the lexicographic order, for which  $W1[\pi ]=0$.

\item \hspace{1cm}   $W1[\pi ]:=1$;

\item \hspace{1cm}  $N:=N+1$;

\item \hspace{1cm} For $i=1,2,\ldots ,m$  \textbf{do}

\hspace{2cm}		\{ \textbf{Begin} of loop 2

\item \hspace{2cm} $\pi = f_r^i (\pi )$.

\item \hspace{2cm} For $j=1,2,\ldots , ,n$  \textbf{do}

\hspace{3cm} \{ \textbf{Begin} of loop 3

\item \hspace{3cm} $\pi :=f_c^j (\pi )$;

\item \hspace{3cm} $W2[\pi ]:=1$;

\hspace{3cm}	End of loop 3\}

\hspace{2cm}	End of loop 2\}

\hspace{1cm}	End of loop 1\}

\item \textbf{End} of the algorithm.
\end{enumerate}
\end{algorithm}

\section{CONCLUSIONS}

Applying the above ideas, a computer program that receives a computer program that gets only one representative from each equivalence class of the factor-set  $\widetilde{B}_{n\times n} =B_{n\times n/\rho}$. The purpose of these calculations was to describe and classify some textile structures [\cite{textile}]. The results relate to obtaining quantitative estimation of all kinds of textile fabric.

In fact, the cardinality of the factor-set M coincides with an integer sequence noted in On-Line Encyclopedia of Integer Sequences [\cite{A179043}] as number A179043, namely

\begin{center}
A179043=\{ 2, 7, 64, 4156, 1342208, 1908897152, 11488774559744, 288230376353050816, 29850020237398264483840, 12676506002282327791964489728, 21970710674130840874443091905462272, 154866286100907105149651981766316633972736, ... \}
\end{center}

\end{document}